\documentclass[12pt]{amsart}

\usepackage{amssymb, amsmath, amscd, eucal, rotating}

\usepackage{tikz-cd}
\usepackage{hyperref}
\input xy
 \xyoption{all}

\usepackage{fullpage}

\numberwithin{equation}{section}

\newcommand{\beq}{\begin{equation}}
\newcommand{\eeq}{\end{equation}}

\theoremstyle{plain}

\newcommand{\fg}{\mathfrak g}
\newcommand{\TT}{{\mathbb T}}

\newcommand\cS{\mathcal{S}}

\newtheorem{proposition}{Proposition}[section]
\newtheorem{theorem}[proposition]{Theorem}		
\newtheorem*{theorem*}{Theorem}		

\newtheorem{lemma}[proposition]{Lemma}

\theoremstyle{definition}

\newtheorem{remark}[proposition]{Remark}

\theoremstyle{remark}

\newcommand{\RR}{\mathbb R}

\newcommand{\ZZ}{\mathbb Z}

\newcommand{\CC}{\mathbb C}

\newcommand{\cF}{\mathcal F}
\newcommand{\cK}{\mathcal K}
\newcommand{\cE}{\mathcal E}
\newcommand{\cA}{\mathcal A}

\newcommand{\be}{\begin{equation}}
\newcommand{\ee}{\end{equation}}

\newcommand{\dbar}{\bar\partial}

\DeclareMathOperator{\im}{im}
\DeclareMathOperator{\tr}{tr}

\DeclareMathOperator{\vol}{vol}

\renewcommand{\Im}{\mathsf{Im}}

\DeclareMathOperator{\Hom}{Hom}
\DeclareMathOperator{\coker}{coker}

\begin{document}

\title[Fractional Quantum Numbers, Complex Orbifolds and Noncommutative Geometry]
{Fractional Quantum Numbers, Complex Orbifolds and Noncommutative Geometry}

	\author{Varghese Mathai}
	\address{School of Mathematical Sciences, University of Adelaide, Adelaide SA 5005 Australia. }
	\email{mathai.varghese@adelaide.edu.au}
	
\author{Graeme Wilkin}
\address{Department of Mathematics, University of York, YO10 5DD, UK}
\email{graeme.wilkin@york.ac.uk}

\begin{abstract}
This paper studies the conductance on the universal homology covering space $Z$
of $2$D orbifolds in a strong magnetic field, thereby removing the integrality constraint on the magnetic field in earlier works
\cite{ASZ94,PrietoCMP,MathaiWilkin19} in the literature.  We consider a natural Landau Hamiltonian on $Z$ and 
study its spectrum which we prove consists of a finite number of low-lying isolated points and calculate the 
von Neumann degree of the associated
holomorphic spectral orbibundles when the magnetic field $B$ is large, and obtain
fractional quantum numbers as the conductance. 
\end{abstract}


\maketitle

\date{\today}

\tableofcontents

\maketitle

\thispagestyle{empty}

\baselineskip=16pt

\section*{Introduction}

The fundamental work of Avron, Seiler, Zograf \cite{ASZ94} studies a class of  quantum systems on compact Riemann
surfaces for which the transport
coefficients simultaneously display quantization and fluctuation.
The two related notions of transport coefficients
that they consider are {\em conductance} and {\em charge transport}.
The mathematical tool used in their work is a local families index theorem due to Quillen
\cite{Quillen85, A-GMV86}. It splits the conductances
into two parts. The first is explicit and  universal, that is, it 
is, up to an integral  factor, the canonical symplectic form on the space of
Aharonov-Bohm fluxes and is quantized,  therefore providing a
connection to the Integer Quantum Hall effect (IQHE) \cite{ TKNN, ASY, Bellissard}.
 The second piece in the formula is a complete
derivative, hence it does not affect charge transport. It affects however
the conductance as a fluctuation term. In contrast to the first, it
depends on spectral properties of the Hamiltonian, as it is related
to the zeta function regularization of its determinant. In \cite{ASZ94}, Avron, Seiler and Zograf study
transport coefficients associated to the ground state for any compact Riemann surface $X$
where the magnetic field $B$ is a large integer. In \cite{PrietoCMP} the 
transport coefficients associated to the eigensections of low lying eigenvalues
are considered for any compact Riemann surface where the magnetic field $B$ is a large integer. 
In earlier work \cite{MathaiWilkin19}, the
transport coefficients associated to the eigensections of low lying eigenvalues
are considered for any compact complex 2D orbifold where the magnetic field $B$ is a large fraction. 

The goal of this paper is to remove the topological constraints on the magnetic field (either integrality or rationality) in earlier works
\cite{ASZ94,PrietoCMP,MathaiWilkin19} in the literature. To achieve this, we instead study the conductance and charge transport on the universal homology orbi-covering space $Z$
of $2$D orbifolds in a strong magnetic field. It turns out that there is a projective unitary action of $\ZZ^{2g}$ on $L^2(Z)$, known as magnetic translations, which commutes 
with the self-adjoint Hamiltonian
$$
H = \frac{\hbar^2}{2m}\left(\nabla_A^*\nabla_A +\frac{R}{6}\right) ,
$$
where $\nabla_A = d + i A$ is a connection on the trivial line bundle on $Z$ with curvature $(\nabla_A)^2= i \tilde B,$ and
$R$ is the constant scalar curvature of $Z$ in a hyperbolic metric. We also assume that the magnetic field $\tilde B$ is a constant multiple of the volume form $\theta\,d\vol_Z$ for some large value of
$\theta \in \RR$.
 Since $Z$ is a noncompact Riemann surface, the Hamiltonian $H$ acting
on $L^2(Z)$ typically has continuous spectrum. However, we show in Theorem \ref{thm:main-results} that there are finitely many eigenvalues $\{\mu_1, \ldots \mu_{m-1}\}$ of the Hamiltonian $H$ 
that are isolated in the spectrum of $H$ and which are near zero. Let $\{E_{\mu_j}, \, j=1,\ldots m-1\}$ denote the corresponding eigenspaces. In Theorem \ref{thm:main-results} we show that these eigenspaces consist of
holomorphic sections and compute the rank $r_j = \dim_\mathbb{C} E_{\mu_j}$. Let $\{P_j, \, j=1,\ldots m-1\}$ denote the respective orthogonal projections $P_j:L^2(Z) \to E_{\mu_j}$ and let $\tau$ denote the von Neumann trace. Then 
by \cite{MathaiRosenberg20}, $E_{\mu_j}$ has a $\dbar$-operator $\overline\nabla = P_j(\dbar)^{r_j}$, where $P_j \in M(r_j, A_\theta)$ and $A_\theta$ is a complex noncommutative torus in dimension $2g$, generated by the magnetic translations.
An open question is whether there is a connection on $E_{\mu_j}$ such that  $\overline\nabla^2=0$ that is  $\overline\nabla$ is a {\em flat $\dbar$ operator} in the sense of  \cite{MathaiRosenberg20}, which is known to be true in the commutative case as shown in \cite{PrietoCMP}. We remark that the Hamiltonian $H$ can be interpreted as a 
family of Hamiltonians parametrised by $A_\theta$
 which can be viewed as the noncommutative analog of the Jacobian variety of $X$.

Hyperbolic 2D surfaces
are typically prohibited experimentally as they cannot be isometrically embedded in three dimensional Euclidean space \cite{Hilbert01}. 
A recent striking development in \cite{KFH19} indisputably shows that lattices of certain resonators can be used to produce artificial photonic materials
 in an effectively curved space, including the 2D hyperbolic plane. 
  In particular, they conducted numerical tight-binding simulations of hyperbolic analogs of the Kagome lattice and demonstrated that they display a flat band,
  similar to that of their Euclidean counterpart. The authors of \cite{KFH19} also constructed a proof-of-principle experimental device which realizes a finite section of non-interacting heptagon-kagome lattice. 

We mention alternate approaches to (fractional) quantum numbers on hyperbolic space. These approaches are for smooth surfaces \cite{CHMM98, CHM99, CHM06}, for orbifolds \cite{MathaiMarcolli99, MathaiMarcolli01, MathaiMarcolli06}
for the bulk-boundary correspondence \cite{MathaiThiang19} and for orbifold symmetric products \cite{MarcolliSeipp17}. These papers use operator algebras and noncommutative geometry methods, 
in contrast to the holomorphic geometry methods used in this paper.
For a recent analysis of the IQHE, see \cite{KMMW}, where the
generating functional, the adiabatic curvature and the adiabatic phase for the IQHE 
are studied on a compact Riemann surface with integral magnetic field, but using holomorphic methods inspired by  \cite{ASZ94}.
\\

\noindent{\bf Acknowledgements.} VM thanks the Australian Research Council  for support
via the Australian Laureate Fellowship FL170100020. 
He gave a talk partly based on this paper at the conference, {\em Topological Phases of Interacting Quantum Systems}, BIRS, Oaxaca, 
Mexico, June 2--7, 2019. 
GW would also like to thank the University of Adelaide for their hospitality during the development of this paper. His visit was funded by FL170100020.

\section{Preliminaries}

\subsection{The maximal abelian cover}

Let $X$ be a compact Riemann surface of genus $g$ with a hyperbolic metric. The first homology is the abelianisation of $G := \pi_1(X)$
\begin{equation*}
H_1(X) = G / [G, G] .
\end{equation*}
Let $p : Z \rightarrow X$ be the maximal abelian cover with $\pi_1(Z) = [G, G]$. For any $x \in X$, we have $p^{-1}(x) \cong H_1(X)$, and so the commutator subgroup has infinite index.

A theorem of Griffiths \cite[4.2]{Griffiths63} shows that the commutator subgroup $[G,G]$ is then a free group. The homology $H_1(Z)$ is then the corresponding free abelian group, and hence the cohomology with coefficients in any ring $R$ is the dual $H^1(Z, R) \cong \Hom(H_1(Z), R)$. Putman \cite[Lem. A.1]{Putman07} shows that the generators of the commutator subgroup can be realised geometrically as loops on the surface that bound a one-holed torus.

The infinite cover $p : Z \rightarrow X$ defined above is induced from the Abel-Jacobi inclusion $X \hookrightarrow J(X)$ and the universal cover of the Jacobian
\begin{equation}\label{diag:abelian}
\begin{tikzcd}
Z \arrow{r}{\tilde\iota} \arrow{d}{p} & \widetilde{J(X)} \arrow{d}{q} \\
X \arrow{r}{\iota} & J(X)
\end{tikzcd}
\end{equation}
To see this, let $q : Z' \rightarrow X$ be the covering of $X$ induced from the universal cover of the Jacobian. Note that $Z'$ is connected. Then choose $x \in X$ and note that the action of $H_1(X) \cong H_1(J(X)) \cong \pi_X(J(X))$ on $p^{-1}(x) \cong H_1(X) \cong q^{-1}(x)$ is the same for both $Z$ and $Z'$. Therefore the two coverings are isomorphic.

\subsection{Magnetic translations}
The Abel-Jacobi inclusion $\iota: X \hookrightarrow J(X)$ induces an isomorphism $\iota^*: H^1(J(X), \mathbb{R})\cong H^1(X, \mathbb{R}) $, and the induced map 
$\iota^*: H^2(J(X), \mathbb{R})\twoheadrightarrow H^2(X, \mathbb{R}) $ is surjective. In fact, if the magnetic field is 
$B= \theta \omega$, where $\omega$ is the K\"ahler 2-form on $X$ and $\theta \in \RR$, then one has,
$$
[B]= [\theta \omega] = \iota^*\left[\frac{\theta}{g}\, \Theta_X\right] \in H^{1,1}(X, \RR)
$$
where $[\Theta_X] \in H^{1,1}(J(X), \RR)$ is the Theta divisor.

Let $B' = \frac{\theta}{g}\, \Theta_X$be the closed (1,1)-form on the Jacobian $J(X)$. Then $\tilde B'= q^*B'$ is a closed 2-form on $ \widetilde{J(X)}$. Since $ \widetilde{J(X)}$ is 
contractible, $H^2( \widetilde{J(X)}, \RR)=\{0\}$, 
so that $\tilde B'=dA'$ for some 1-form $A'$ on $\widetilde J(X)$.

Now $0=\gamma^*(\tilde B') - \tilde B'= d(\gamma^*A' - A')$, so that $\gamma^*A' - A'$ is a closed 1-form on $ \widetilde{J(X)}$ for all $\gamma \in \Gamma =H_1(X, \ZZ) =  H_1(J(X), \ZZ)$.
Since $ \widetilde{J(X)}$ is contractible, $H^1( \widetilde{J(X)}, \RR)=\{0\}$, so that $\gamma^*A' - A' = d\phi'_\gamma$, where $\phi'_\gamma$ is a smooth function on 
$\widetilde{J(X)}$ normalised by $\phi'_\gamma(\tilde{\iota}(x_0))=0$ for all $\gamma \in \Gamma$ and for some $x_0 \in Z$.

Define $A := \tilde\iota^*A' \in \Omega^1(Z)$ and $\phi_\gamma := \tilde\iota^*\phi'_\gamma$, and note that $\tilde B = p^*B = dA$. Moreover, $\phi_\gamma$ satisfies the identity
\begin{equation}\label{eqn:basic-phi-identity}
\phi_{\gamma_2}(x) + \phi_{\gamma_1}(\gamma_2 \cdot x) - \phi_{\gamma_1+\gamma_2}(x) = \phi_{\gamma_1}(\gamma_2 \cdot x_0) 
\end{equation}
for all $x \in Z$ and $\gamma_1, \gamma_2 \in \Gamma$. Now define $\sigma : \Gamma \times \Gamma \rightarrow U(1)$ by
\begin{equation}\label{eqn:def-sigma}
\sigma(\gamma_1, \gamma_2) = e^{i \phi_{\gamma_1}(\gamma_2 \cdot x_0)} .
\end{equation}
Then the above identity \eqref{eqn:basic-phi-identity} shows that $\sigma$ satisfies the cocycle condition
\begin{equation}\label{eqn:sigma-cocycle}
\sigma(\gamma_1, \gamma_2)\sigma(\gamma_1 + \gamma_2, \gamma_3)
= \sigma(\gamma_1, \gamma_2 + \gamma_3)\sigma( \gamma_2, \gamma_3), \quad \gamma_1, \gamma_2, \gamma_3\in \Gamma .
\end{equation}
For each $\gamma \in \Gamma$, define operators on $L^2(Z)$ by
\begin{align*}
U_\gamma f(x) & = f(\gamma^{-1}x)\\
S_\gamma f(x) & = e^{i \phi_\gamma(x)} f(x) .
\end{align*}
Then the above identities imply that $T_\gamma := U_\gamma \circ S_\gamma$ satisfies $T_{\gamma_1} T_{\gamma_2} = \sigma(\gamma_1, \gamma_2) T_{\gamma_1\gamma_2}$, and hence defines a twisted or projective action of $\Gamma$ on $L^2(Z)$.

\subsection{The magnetic Schr\"odinger operator}\label{subsec:magnetic-schrodinger}
Let $\nabla_A = d + i A$ be a connection on the trivial line bundle on $Z$. Then the curvature of $\nabla_A$ is $(\nabla_A)^2= i \tilde B$. Consider the magnetic Schr\"odinger operator
\begin{equation}\label{eqn:magnetic-Schrodinger}
H = \frac{\hbar^2}{2m}\left(\nabla_A^*\nabla_A +\frac{R}{6}\right) ,
\end{equation}
where $R$ is the constant scalar curvature of $Z$. Then $U_\gamma \nabla_A = \nabla_{\gamma^{-1}*A} U_\gamma$ and 
$S_\gamma \nabla_{\gamma^{-1}*A} = \nabla_A S_\gamma$, so that $T_\gamma H = H T_\gamma$ for all $\gamma \in \Gamma$.

Notice that $H$ is bounded below, so the spectral projection $$P_\lambda = \chi_{(-\infty, \lambda)}(H) \in W^*(\Gamma, \sigma) \otimes B(L^2(\cF)),$$
where $\cF$ is a connected fundamental domain for the action of $\Gamma$ on $Z$, $ B(L^2(\cF))$ denotes the bounded operators on the Hilbert space
 $L^2(\cF)$, and $W^*(\Gamma, \sigma) $ denotes the $\sigma$-twisted group von Neumann algebra of $\Gamma$ generated by the magnetic translations acting on 
 $\ell^2(\Gamma)$.
 
If however $\lambda$ is in a spectral gap of $H$, then it is a standard result (cf. \cite[Thm. 1]{BruningSunada92}) that the spectral projection belongs to the much smaller algebra,
 $$P_\lambda = \chi_{(-\infty, \lambda)}(H) = f(H)\in C^*(\Gamma, \sigma) \otimes \cK(L^2(\cF)),$$
where $ \cK(L^2(\cF))$ denotes the compact operators on the Hilbert space
 $L^2(\cF)$, and $C^*(\Gamma, \sigma) $ denotes the $\sigma$-twisted group $C^*$-algebra of $\Gamma$
 and $f$ is a holomorphic function defined in a neighbourhood of $(-\infty, \lambda]$.
 When $\tilde B$ is a constant multiple $\theta\, d\vol_Z$ of the hyperbolic volume, which is the case that we will focus on in this paper,
 then $C^*(\Gamma, \sigma) $  is the noncommutative torus $A_\theta$
in dimension $2g$.
 \\

 \subsection{The noncommutative torus}\label{subsec:noncommutative-torus}
 
 Here we recall the definition of the higher dimensional noncommutative torus, its complex structure and K-theory, and the range of the trace and also the 2-trace on K-theory.
  Let $p=2g$ and $\Theta$ be a $(p\times p)$ skew-symmetric matrix.  
 Then $$\sigma(\gamma, \gamma') = \exp\left(2\pi \sqrt{-1} \sum_{j<k} \Theta_{jk} \gamma_j \gamma'_k\right), \qquad \text{where}\quad \ \gamma=(\gamma_1, \ldots, \gamma_p), 
  \gamma'=(\gamma_1', \ldots, \gamma_p') \in \ZZ^p.$$
 is a $U(1)$-valued group 2-cocycle on $\ZZ^p$. Let $\CC(\ZZ^p, \sigma)$ denote the twisted group algebra, that is for functions $f_1, f_2 \colon: \ZZ^p \to \CC$ of finite support, the twisted convolution product is $$f_1\star_\sigma f_2(\gamma) = \sum_{\gamma_1 + \gamma_2 = \gamma} f_1(\gamma_1) f_2(\gamma_2) \sigma(\gamma_1, \gamma_2)$$
 Then $\CC(\ZZ^p, \sigma)$ acts on bounded operators on $\ell^2(\ZZ^p)$ by the formula above. The operator norm closure $\overline{\CC(\ZZ^p, \sigma)}$ is defined to be the noncommutative torus $A_\Theta$ or the twisted group $C^*$-algebra $C^*(\ZZ^p, \sigma)$.
 
 There is an abstract definition of $A_\Theta$ that is useful to recall.  The 
\emph{noncommutative torus} $A_\Theta$ is the universal $C^*$-algebra with 
$p$ unitary generators $U_j$, $1\le j\le p$, subject to the basic commutation
relation
\[
U_jU_k = e^{2\pi i \Theta_{jk}} U_kU_j.
\]
This algebra carries a \emph{gauge action} of the torus $\TT^p$ via 
\[
t\cdot \left(U_1^{n_1}\cdots U_d^{n_d} \right)
= t_1^{n_1}\cdots t_d^{n_d} U_1^{n_1}\cdots U_d^{n_d} , \quad
t = (t_1, \cdots, t_d)\in \TT^p.
\]
There are associated infinitesimal generators $\delta_j$, 
which are $*$-derivations, defined by
$$
\delta_j(U_k) = 2\pi i \delta_{jk} U_k .
$$
The algebra $A_\Theta$ also carries a canonical tracial state $\tau$ invariant 
under the gauge action, sending $1$ to $1$ and sending a monomial
$U_1^{n_1}\cdots U_d^{n_p}$ to $0$ unless all of the $n_j$ vanish.
Because of the commutation relation, any element of
$A_\Theta$ has a canonical (formal) expansion in terms of
the monomials $U_1^{n_1}\cdots U_d^{n_p}$.  Since every element of $A_\Theta$ 
has a unique expression $ \sum a_{\vec{n}} \, U_1^{n_1}\cdots U_d^{n_p}$, so 
$$
\tau\left(\sum_{\vec{n}} a_{\vec{n}} \, U_1^{n_1}\cdots U_d^{n_p}\right) = a_{\vec{0}}
$$
The {\em smooth noncommutative torus} $\cA_\Theta$ consists of all elements 
$\sum a_{\vec{n}} \, U_1^{n_1}\cdots U_d^{n_p}$ such that the coefficients 
$a_{\vec{n}}$ form a rapidly decreasing sequence in $\cS(\ZZ^p)$. It is also the 
smooth vectors for the gauge action of $\TT^p$ on $A_\Theta$. It follows that 
$\cA_\Theta$ is the domain of the (powers of) derivations $\delta_j^k$ that are the infinitesmal generators of the gauge action. 
There is a natural continuous cyclic 2-cocycle $\tau_c$ on $\cA_\Theta$. Let $f_0, f_1, f_2 \in \cA_\Theta$. Then
\begin{equation}\label{eqn:continuous-cocycle-def}
\tau_c(f_0, f_1, f_2) = \sum_{i=1}^g \tau(f_0(\delta_i(f_1)\delta_{i+g}(f_2) - (\delta_{i+g}(f_1)\delta_{i}(f_2))).
\end{equation}
where $c$ is the area 2-cocycle on $\ZZ^p$ corresponding to the symplectic form. 
 
 The inclusion $\cA_\Theta \hookrightarrow A_\Theta$ is known to induce an isomorphism in K-theory, $K_\bullet(A_\Theta) \cong K_\bullet(\cA_\Theta)$. The range of the trace on K-theory has been computed \cite{BenameurMathai18, Elliott80}:
 \be\label{rangeoftrace}
\tau (K_0(A_\Theta) )= \ZZ +\sum_{0<|I|<p} {\rm Pf}(\Theta_I) \ZZ  + {\rm Pf}(\Theta)\ZZ,
\ee
where $I$ runs over subsets of $\{1,\ldots,p\}$ with an even number of elements, and $\Theta_I$ denotes the skew-symmetric submatrix of 
$\Theta=(\Theta_{ij})$ with $i,j \in I$. The formula (see section 1 in \cite{MathaiQuillen86})
$$
e^{\frac{1}{2}dx^t \Theta dx} = \sum_I {\rm Pf}(\Theta_I) dx^I
$$
is key to this computation, together with the twisted L$^2$-index theorem \cite{Mathai99}.

 We can also compute the range of a certain higher trace on K-theory. Let $\Theta_X$ denote the theta divisor. In real coordinates,
 $\Theta_X = \sum_{i=1}^g dx_i\wedge dx_{i+g}$. Let $c$ be the group cocycle that corresponds to $\Theta_X $. Then setting 
 $I_i = \{i, i+g\}$, one has
 \be\label{rangeofhighertrace}
\tau_c (K_0(\cA_\Theta) )= \sum_{i=1}^g \sum_{I_i\subset I} {\rm Pf}(\Theta_{I\setminus I_i}) \ZZ  
\ee
where $I$ runs over subsets of $\{1,\ldots,p\}$ with an even number of elements, and $\Theta_{I\setminus I_i}$ denotes the skew-symmetric submatrix of $\Theta=(\Theta_{ij})$ with $i,j \in I\setminus I_i$. 
The method of proof again uses \cite{BenameurMathai18,MathaiMarcolli01} and can be deduced from Corollary 5.7.2 in \cite{ProdanSchulz-Baldes16}.

 \subsection{Noncommutative complex torus and holomorphic vector bundles} \label{complexnctori}
 The following is recalled from \cite{MathaiRosenberg20}.
 Define a tangent space of $\cA_\Theta$ to be the
(commutative) Lie algebra $\fg = \text{span}\,(\delta_1,\cdots, \delta_{2n})$.
A \emph{complex structure} on $A_\Theta$ and $\cA_\Theta$
is a choice of an endomorphism $J$ of $\fg$ satisfying $J^2 = -1$.
It thus defines an isomorphism $\fg_\CC\cong \fg^{\text{hol}}
\oplus \fg^{\text{antihol}}$
as a direct sum of holomorphic and antiholomorphic tangent spaces,
namely the $\pm i$-eigenspaces of $J$.  There is a similar splitting
of the complexified cotangent space $\fg_\CC^*$.
The pair $(A_\Theta, J)$ will be called a \emph{{noncommutative} complex torus}
of complex dimension $n$.

A \emph{vector bundle} $E$ over $\cA_\Theta$ will mean a finitely generated
projective (right) module.  
A \emph{holomorphic vector bundle} $E$ will mean such a bundle
equipped with a \emph{holomorphic connection} $\overline\nabla$,
meaning a map $E\to E \otimes \left(\fg^{\text{antihol}}\right)^*$
satisfying the Leibniz rule
\begin{equation}
\label{eq:Leibniz}
\overline\nabla_{\bar\partial_j}(e\cdot a) =
\overline\nabla_{\bar\partial_j}(e)\cdot a +
e \cdot \bar\partial_j(a).
\end{equation}
Note that any vector bundle of rank $r$ can be equipped with a
holomorphic connection simply by writing
$E = p (\cA_\Theta)^r$ for some projection $p$, and then
defining $\overline\nabla = p(\bar\partial)^r$.

More interesting are \emph{flat holomorphic connections}, which 
satisfy the flatness condition $(\overline\nabla)^2=0$.
It is {\em not} the case that every vector bundle has a 
flat holomorphic connection.

\section{Summary of results}
 
This section contains the main results of the paper.
 
 \begin{theorem}\label{thm:main-results}
 
 Let $H$ be the magnetic Schr\"odinger operator of \eqref{eqn:magnetic-Schrodinger}.
 
\begin{enumerate}

\item \label{item:discrete-spectrum} Let $m>0$ be the largest integer such that $\theta - m (2g-2) > 0$. 
If $\lambda = \mu_q$ for some integer $q$ such that $0\le q<m$, where 
  \begin{equation*}
  \mu_q = (2q+1) \theta- q(q+1) (2g-2),
  \end{equation*}
  then $\lambda = \mu_q$ is an isolated point in the spectrum of $H$.
It follows that the spectral projection $P_\lambda \in A_\theta \otimes \cK(L^2(\cF)).$

\bigskip

\item \label{item:holomorphic-eigensections} The eigenspace $E_\lambda = \Im (P_\lambda)$ for $\lambda$ as above, consists of holomorphic sections. 

\bigskip

\item \label{item:von-Neumann-dimension} The von Neumann dimension of $E_\lambda$, $\dim_\tau(E_\lambda)$, is equal to $\tau(P_\lambda)$, where $\tau:  A_\theta\otimes \cK(L^2(\cF)) \to \CC$ is the von Neumann trace. Then $E_{\mu_q}$ is infinite dimensional since we show that
\begin{equation*}
\dim_\tau(E_{\mu_q}) =( 2 q + 1) (1-g) + \theta >0 ,
\end{equation*}
and so the von Neumann dimension $\dim_\tau(E_{\mu_q})$ of the spectral subspace grows linearly with $q$.

\bigskip

\item \label{item:chern-number} Let $\lambda=\mu_q$. The Chern number of $E_\lambda$ is
\begin{equation}\label{eqn:chern-number-eigenspace}
\tau_2(P_\lambda, P_\lambda, P_\lambda)=\tau(P_\lambda dP_\lambda dP_\lambda) = 2g .
\end{equation}
In particular, it is an integer and it is independent of $q$. In the case of an orbifold given as a quotient $X = X' / \Gamma$ of a smooth Riemann surface $X'$ by a finite group $\Gamma$, the Chern number is a rational number
\begin{equation}\label{eqn:chern-number-orbifold}
2g-2 + \# (R / \Gamma) + \frac{2-n}{|\Gamma|} ,
\end{equation}
where $R$ is the ramification divisor of the ramified cover $X' \rightarrow X$. Here we use the well known fact that the noncommutative torus $C^*(\Gamma, \sigma)$ comes with a canonical cyclic 2-cocycle, $$\tau_2(f_0, f_1, f_2)=\sum_{j=1}^g \tau(f_0 \delta_jf_1 \delta_{j+g} f_2 - \delta_{j+g}f_1 \delta_{j} f_2 ))$$
 for $f_k \in \CC(\Gamma, \sigma)$. This is derived to be the conductance 2-cocycle $\tau_K$ in Corollary 5, \cite{CHMM98}, see also page 73 in \cite{MathaiMarcolli01}.
 
\end{enumerate}

\end{theorem}

The subsequent sections contain the proof of these results. Parts \eqref{item:discrete-spectrum} and \eqref{item:holomorphic-eigensections} follow from Theorem \ref{thm:discrete-spectrum}, part \eqref{item:von-Neumann-dimension} is proved in Lemma \ref{dim} and the results of part \eqref{item:chern-number} are contained in \eqref{eqn:chern-number-smooth-proof} and \eqref{eqn:chern-number-orbifold-proof}.

\section{Discrete values of the spectrum of the Laplacian on the maximal abelian cover}

The main result of this section is Theorem \ref{thm:discrete-spectrum}, proving the first two parts of Theorem \ref{thm:main-results}.

Let $\mathcal{E} \rightarrow X$ be a complex line bundle. Since $Z$ is a noncompact surface then the pullback $\tilde{\mathcal{E}}$ is topologically trivial, and any holomorphic structure on $\tilde{\mathcal{E}}$ is also trivial (see for example \cite[Thm. 30.4]{Forster81}), however $\tilde{\mathcal{E}}$ admits many different gauge-equivalent structures as a holomorphic bundle, or equivalently as a Hermitian bundle with a unitary connection. In order to normalise the eigenvalues of the Laplacian, we will fix a metric on $X$ of constant Gauss curvature $\chi(X)$ so that $\vol(X) = 2\pi$, and use the pullback metric on the maximal abelian cover $Z \rightarrow X$. Fix a Hermitian metric on $\tilde{\mathcal{E}}$ and a Hermitian connection $\nabla^\theta$, and let $\theta = i*F_{\nabla^\theta} \in \mathbb{R}$ be the curvature, which we assume from now on to be constant. In general $\theta$ can be any real number; in the special case that $\nabla^\theta$ is the pullback of a constant curvature connection on a line bundle $\mathcal{E} \rightarrow X$, then the Chern-Weil formula $\deg(\mathcal{E}) = \frac{i}{2\pi} \int_X * F_{\nabla^\theta} = \theta$ shows that $\theta = \deg(\mathcal{E}) \in \mathbb{Z}$.

Elliptic regularity shows that the eigensections of $(\nabla^\theta)^* \nabla^\theta$ are smooth. In the following, we will use the Sobolev space $H^2(Z, \mathcal{E}) \subset L^2(Z, \mathcal{E})$ as the domain of the Laplacian, on which $(\nabla^\theta)^* \nabla^\theta$ is self-adjoint, and continuous as an operator $H^2(Z, \mathcal{E}) \rightarrow L^2(Z, \mathcal{E})$.

The main result of this section is Theorem \ref{thm:discrete-spectrum}, which shows that the low-lying eigenvalues of the Laplacian $(\nabla^\theta)^* \nabla^\theta$ are the discrete values given by \eqref{eqn:explicit-eigenvalues}, and that the corresponding eigensections are images of holomorphic sections of the associated bundle $K^{-q} \otimes \tilde{\mathcal{E}}$ given by \eqref{eqn:explicit-eigensections}.

Decomposing $\nabla^\theta$ into $(0,1)$ and $(1,0)$ parts induces operators $\partial^{\nabla^\theta}$ and $\bar{\partial}^{\nabla^\theta}$, which satisfy the following identities for the Laplacian and the curvature
\begin{equation}\label{eqn:decompose-laplacian}
(\nabla^\theta)^* \nabla^\theta = \Delta_{\bar{\partial}}^\theta + \Delta_{\partial}^\theta, \quad i * F_{\nabla^\theta} = \Delta_{\partial}^\theta - \Delta_{\bar{\partial}}^\theta .
\end{equation}
Combining these shows that
\begin{equation}\label{eqn:Laplacian-curvature}
(\nabla^\theta)^* \nabla^\theta = 2 \Delta_{\bar{\partial}}^\theta + i*F_{\nabla^\theta} = 2 \Delta_{\bar{\partial}}^\theta + \theta ,
\end{equation}
and so the eigensections of $(\nabla^\theta)^* \nabla^\theta$ with eigenvalue $\mu$ correspond exactly to eigensections of $\Delta_{\bar{\partial}}^\theta$ with eigenvalue $\frac{1}{2} (\mu - \theta)$. The first consequence of this is that $\mu \geq \theta$, since $\Delta_{\bar{\partial}}^\theta$ is non-negative. Secondly, we see that the sections in the kernel of $\Delta_{\bar{\partial}}^\theta$ (corresponding to the holomorphic sections of $\mathcal{E}$) are eigensections of $(\nabla^\theta)^* \nabla^\theta$ with eigenvalue $\theta$.

The basic example of such a connection on a trivial bundle with trivial metric on $\mathbb{C}^g$ is
\begin{equation*}
\bar{\partial}^{\nabla^\theta} = \bar{\partial} + \frac{1}{4} i \theta z d \bar{z}, \quad \partial^{\nabla^\theta} = \partial - \frac{1}{4} i \bar{z} dz ,
\end{equation*}
where we use the shorthand $z d\bar{z} := \sum_{k=1}^g z_k d \bar{z}_k$, $\bar{z} dz := \sum_{k=1}^g \bar{z}_k dz_k$. One can easily check that $e^{-\frac{1}{4} i \theta |z|^2}$ is in $\ker \bar{\partial}^{\nabla^\theta}$ and therefore is an eigenfunction of $(\nabla^\theta)^* \nabla^\theta$ with eigenvalue $\theta$. Since the inclusion $Z \hookrightarrow \mathbb{C}^g$ induced from the Abel-Jacobi map is a holomorphic embedding, then the restriction of $e^{-\frac{1}{4} i \theta |z|^2}$ to $Z \subset \mathbb{C}^g$ is also in $\ker \bar{\partial}^{\nabla^\theta}$ and therefore an eigenfunction of the Laplacian.

Therefore the lowest eigenvalue of $(\nabla^\theta)^* \nabla^\theta$ is determined by the curvature $\theta$, and the corresponding eigensections are determined by the holomorphic sections of $\mathcal{E}$. The next theorem extends this result to show that the higher eigenvalues of $(\nabla^\theta)^* \nabla^\theta$ also have an explicit description. Most importantly, they are discrete in the interval $[0, \mu_m)$, where $m$ and $\mu_m$ are defined below. For compact surfaces, this result is due to Prieto in \cite{PrietoDGA}, and the proof below involves extending these techniques to apply to the noncompact infinite genus surface $Z$. 

First, we set some notation. Let $T$ be the tangent bundle of $Z$, which we assumed to have constant Gauss curvature $\chi(X)$. The bundle $T$ also has a canonical holomorphic structure and Hermitian metric induced from that of $Z$, and the dual is identified with the canonical bundle $K$, which then has an induced connection. Define 
\begin{equation*}
\nabla^{\theta, q} := \tilde\nabla^{K^{q}} \otimes \nabla^\theta
\end{equation*}
to be the connection on $\tilde{\mathcal{E}} \otimes K^q$ for any $q \in \mathbb{Z}$, and decompose into $(1,0)$ and $(0,1)$ parts to define the associated operators
\begin{align*}
\partial^{\nabla^{\theta, q}} : \Omega^0(\tilde{\mathcal{E}} \otimes K^q) & \rightarrow \Omega^{1,0}(\tilde{\mathcal{E}} \otimes K^q) \cong \Omega^0(\tilde{\mathcal{E}} \otimes K^{q+1}) \\
\bar{\partial}^{\nabla^{\theta, q}} : \Omega^0(\tilde{\mathcal{E}} \otimes K^q) & \rightarrow \Omega^{0,1}(\tilde{\mathcal{E}} \otimes K^q) \cong \Omega^0(\tilde{\mathcal{E}} \otimes K^{q-1}) .
\end{align*}

\begin{remark}\label{rem:injective-10}
In analogy with \eqref{eqn:decompose-laplacian}, note that the Laplacians satisfy the identity
\begin{equation}
i * F_{\nabla^{\theta, q}} = \Delta_{\partial^{\theta, q}} - \Delta_{\bar{\partial}^{\theta, q}} .
\end{equation}
In particular, if $\theta - q (2g-2)  > 0$, then $\partial^{\nabla^{\theta,-q}} : \Omega^0(\tilde{\mathcal{E}} \otimes K^{-q}) \rightarrow \Omega^0(\tilde{\mathcal{E}} \otimes K^{-q+1})$ is injective, since $\Delta_{\partial^{\theta,-q}} = \theta - q(2g-2)  + \Delta_{\bar{\partial}^{\theta,-q}}$ is strictly positive.
\end{remark}

We then have two sequences of homomorphisms given by composing these operators as $q$ increases or decreases 
\begin{equation*}
\begin{tikzcd}
\Omega^0(\tilde{\mathcal{E}} \otimes K^{q-1})  \arrow[bend left = 10]{r}{\partial^{\nabla^{\theta, q-1}}} & \Omega^0(\tilde{\mathcal{E}} \otimes K^{q}) \arrow[bend left = 10]{r}{\partial^{\nabla^{\theta, q}}} \arrow[bend left = 10]{l}{\bar{\partial}^{\nabla^{\theta, q}}} & \Omega^0(\tilde{\mathcal{E}} \otimes K^{q+1}) \arrow[bend left = 10]{r}{\partial^{\nabla^{\theta, q+1}}} \arrow[bend left = 10]{l}{\bar{\partial}^{\nabla^{\theta, q+1}}} & \Omega^0(\tilde{\mathcal{E}}  \otimes K^{q+2}) \arrow[bend left = 10]{l}{\bar{\partial}^{\nabla^{\theta, q+2}}}
\end{tikzcd}
\end{equation*}

\begin{theorem}\label{thm:discrete-spectrum}
Fix a metric on $Z$ with constant Gauss curvature $\chi(X) = 2-2g$, and let $m$ be the largest integer such that $\theta - m(2g-2) > 0$. For each integer $0 \leq q \leq m$, define
\begin{equation}\label{eqn:explicit-eigenvalues}
\mu_q = (2q+1) \theta - q(q+1) (2g-2) .
\end{equation}
Then the spectrum of the Laplacian in the interval $[0, \mu_m)$ consists of the discrete eigenvalues $\mu_q$ for each $q = 0, \ldots, m-1$. For each such eigenvalue $\mu_q$, the corresponding space of eigensections $E_{\mu_q} $ is equal to 
\begin{equation}\label{eqn:explicit-eigensections}
E_{\mu_q}  =\overline{\partial^{\nabla^{\theta,-1}} \circ \cdots \circ \partial^{\nabla^{
\theta,-q}} \left( \ker \bar{\partial}^{\nabla^{\theta,-q}} \right)} \subset \Omega^0(\tilde{\mathcal{E}}) .
\end{equation}
In particular, $E_{\mu_q} $ and $ \ker \bar{\partial}^{\nabla^{\theta,-q}}$
are isomorphic as $A_\theta$ modules.

\end{theorem}

\begin{remark}
These are exactly the same as the eigenvalues in the discrete spectrum of the Laplacian on $\mathbb{H}^2$ computed by Comtet and Houston \cite{ComtetHouston85} (in \eqref{eqn:explicit-eigenvalues} $Z$ has Gauss curvature $-\frac{1}{a^2} = \chi(X)$). This is a nontrivial observation, since the cover $\mathbb{H}^2 \rightarrow Z$ has structure group the infinitely generated free group $\pi_1(Z) = [\pi_1(X), \pi_1(X)]$, and so the $L^2$ spectrum on $Z$ is not necessarily contained in the $L^2$ spectrum on $\mathbb{H}^2$. 
\end{remark}

The proof of this theorem is contained in Propositions \ref{prop:eigensections} and \ref{prop:spectral-gaps} below. Before proving these propositions we first need some basic results about the Laplacian on $Z$. In the following, the closure will always be taken in the Sobolev space $H^2(Z, \tilde{\mathcal{E}})$, which we use for the domain of the Laplacian.

\begin{lemma}\label{lem:decompose-sections}
For each $q$, the space of sections is a direct sum
\begin{equation}
\Omega^0(\tilde{\mathcal{E}} \otimes K^q) \cong \ker \bar{\partial}^{\nabla^{\theta,q}} \oplus \overline{\im \partial^{\nabla^{\theta, q-1}}} .
\end{equation}
\end{lemma}

The next result describes how the curvature measures the failure of the operators $\partial^{\nabla^{\theta,-(q+1)}}$ and $\bar{\partial}^{\nabla^{\theta, q + 1}}$ to commute with the associated Laplacians. The proof is a local calculation as carried out in \cite[Prop. 10]{PrietoDGA}, and so it also applies to the maximal abelian cover $Z$.
\begin{lemma}\label{lem:failure-to-commute}
In the following diagram
\begin{equation*}
\begin{tikzcd}
\Omega^0(K^{q+1} \otimes \tilde{\mathcal{E}}) \arrow{r}{\bar{\partial}^{\nabla^{\theta, q+1}}} \arrow{d}{\Delta_{\partial}^{\theta, q+1}} & \Omega^0(K^q \otimes \tilde{\mathcal{E}}) \arrow{d}{\Delta_{\partial}^{\theta, q}} \\
\Omega^0(K^{q+1} \otimes \tilde{\mathcal{E}}) \arrow{r}{\bar{\partial}^{\nabla^{\theta, q+1}}} & \Omega^0(K^q \otimes \tilde{\mathcal{E}})
\end{tikzcd}
\end{equation*}
we have 
\begin{equation}\label{eqn:pos-degree-commutativity}
\Delta_\partial^{\theta, q} \circ \bar{\partial}^{\nabla^{\theta, q+1}} - \bar{\partial}^{\nabla^{\theta, q+1}} \circ \Delta_\partial^{\theta, q+1} = - i \bar{\partial}^{\nabla^{\theta, q+1}} * F_{\nabla^{\theta, q+1}} .
\end{equation}
Equivalently, the above diagram commutes up to a factor of $- i \bar{\partial}^{\nabla^{\theta, q+1}} * F_{\nabla^{\theta, q+1}}$. Similarly, in the following diagram
\begin{equation*}
\begin{tikzcd}
\Omega^0(K^{-(q+1)} \otimes \tilde{\mathcal{E}}) \arrow{r}{\partial^{\nabla^{\theta, -(q+1)}}} \arrow{d}{\Delta_{\bar{\partial}}^{\theta,-(q+1)}} & \Omega^0(K^{-q} \otimes \tilde{\mathcal{E}}) \arrow{d}{\Delta_{\bar{\partial}}^{\theta,-q}} \\
\Omega^0(K^{-(q+1)} \otimes \tilde{\mathcal{E}}) \arrow{r}{\partial^{\nabla^{\theta,-(q+1)}}} & \Omega^0(K^{-q} \otimes \tilde{\mathcal{E}})
\end{tikzcd}
\end{equation*}
we have 
\begin{equation}\label{eqn:neg-degree-commutativity}
\Delta_{\bar{\partial}}^{\theta,-q} \circ \partial^{\nabla^{\theta,-(q+1)}} - \partial^{\nabla^{\theta,-(q+1)}} \circ \Delta_{\bar{\partial}}^{\theta,-(q+1)} =  i \partial^{\nabla^{\theta,-(q+1)}} *F_{\nabla^{\theta,-(q+1)}} .
\end{equation}
Equivalently, the above diagram commutes up to a factor of $+i \partial^{\nabla^{\theta,-(q+1)}} *F_{\nabla^{\theta,-(q+1)}}$.
\end{lemma}

\subsection{An expression for the eigensections}

Recall the definition of $m$ as the largest integer such that $K^{-m} \otimes \tilde{\mathcal{E}}$ has positive curvature, or equivalently such that $\theta - m(2g-2) > 0$.

The first proposition constructs eigensections of the Laplacian with given eigenvalues.

\begin{proposition}\label{prop:eigensections}
For all $0 \leq q \leq m$, the sections in 
\begin{equation}\label{eqn:level-q-eigensections}
\overline{\partial^{\nabla^{\theta,-1}} \circ \cdots \circ \partial^{\nabla^{\theta,-q}} \left( \ker \bar{\partial}^{\nabla^{\theta,-q}} \right)} \subset \Omega^0(\tilde{\mathcal{E}})
\end{equation}
are eigensections of $(\nabla^\theta)^* \nabla^\theta$ with eigenvalue 
\begin{equation*}
\mu_q = (2q+1) \theta - q(q+1) (2g-2).
\end{equation*}
\end{proposition}

\begin{proof}
Let $s_{-q} \in \ker \bar{\partial}^{\nabla^{\theta,-q}} \subset \Omega^0(\tilde{\mathcal{E}} \otimes K^{-q})$, and let $s = \partial^{\nabla^{\theta,-1}} \circ \cdots \circ \partial^{\nabla^{\theta,-q}} s_{-q}$. Then Lemma \ref{lem:failure-to-commute} implies that
\begin{align}\label{eqn:first-commuting-term}
\begin{split}
\Delta_{\bar{\partial}}^\theta s = \Delta_{\bar{\partial}}^\theta (\partial^{\nabla^{\theta,-1}} \circ \cdots \circ \partial^{\nabla^{\theta,-q}} s^{-q} ) & = \partial^{\nabla^{\theta,-1}} \circ \Delta_{\bar{\partial}}^{\theta,-1} \circ \partial^{\nabla^{\theta,-2}} \circ \cdots \circ \partial^{\nabla^{\theta,-q}} s^{-q} \\
 & \quad \quad + \partial^{\nabla^{\theta,-1}} (i * F_{\nabla^{\theta,-1}}) \partial^{\nabla^{\theta,-2}} \circ \cdots \circ \partial^{\nabla^{\theta,-q}} s^{-q} \\
 & = \partial^{\nabla^{\theta,-1}} \circ \Delta_{\bar{\partial}}^{\theta,-1} \circ \partial^{\nabla^{\theta,-2}} \circ \cdots \circ \partial^{\nabla^{\theta,-q}} s^{-q} \\
 & \quad \quad + (\theta - (2g-2)) s ,
\end{split}
\end{align}
since the curvature $i * F_{\nabla^{\theta,-1}}$ of the connection on $\tilde{\mathcal{E}} \otimes K^{-1}$ is constant and equal to $\theta - (2g-2)$. Repeatedly applying Lemma \ref{lem:failure-to-commute} shows that
\begin{align}\label{eqn:all-commuting-terms}
\begin{split}
\Delta_{\bar{\partial}}^\theta (\partial^{\nabla^{\theta,-1}} \circ \cdots \circ \partial^{\nabla^{\theta,-q}} s_{-q} ) & = \partial^{\nabla^{\theta,-1}} \circ \cdots \circ \partial^{\nabla^{\theta,-q}} \circ \Delta_{\bar{\partial}}^{\theta,-q} s_{-q} + \left( \sum_{\ell=1}^q (\theta - \ell (2g-2)) \right) s \\
 & = \left( q \theta - \frac{q(q+1) (2g-2)}{2}  \right) s ,
\end{split}
\end{align}
since $s_{-q}$ is holomorphic by assumption. Therefore \eqref{eqn:Laplacian-curvature} shows that
\begin{equation*}
(\nabla^\theta)^* \nabla^\theta s = 2 \Delta_{\bar{\partial}}^\theta s + i * F_{\nabla^\theta} s = \left( (2q+1)\theta - q(q+1) (2g-2) \right) s .
\end{equation*}
The statement of \eqref{eqn:level-q-eigensections} on the closure of this space of sections in $H^2(Z, \mathcal{E})$ follows by taking sequences of such sections and noting that the Laplacian is continuous as an operator $H^2(Z, \mathcal{E}) \rightarrow L^2(Z, \mathcal{E})$.
\end{proof}

\subsection{Spectral gaps around the low-lying eigenvalues}

In this section we show that the eigensections of Proposition \ref{prop:eigensections} are the only eigensections with eigenvalue bounded above by $\mu_m = (2m+1)\theta - m(m+1) (2g-2)$. As a consequence, we see that in the interval $[0, \mu_m)$, the spectrum is discrete and consists of eigenvalues $\mu_0, \ldots, \mu_{m-1}$. 

\begin{proposition}\label{prop:spectral-gaps}
In the interval $[0,\mu_m)$, the spectrum of the Laplacian $(\nabla^\theta)^* \nabla^\theta$ takes the discrete values $\mu_q = (2q+1) \theta - q(q+1) (2g-2)$ for each $q = 0, \ldots, m-1$, and the space of eigensections is
\begin{equation*}
E_{\mu_q} = \overline{\partial^{\nabla^{\theta,-1}} \circ \cdots \circ \partial^{\nabla^{\theta,-q}} \left( \ker \bar{\partial}^{\nabla^{\theta,-q}} \right)} \subset \Omega^0(\tilde{\mathcal{E}}) .
\end{equation*}
This defines an isomorphism of $A_\theta$-modules $E_{\mu_q} \cong H^0(\tilde{\mathcal{E}} \otimes K^{-q})$.
\end{proposition}

\begin{proof}
From \eqref{eqn:Laplacian-curvature} it follows directly that $s \in \ker \bar{\partial}^{\nabla^\theta}$ implies that $(\nabla^\theta)^* \nabla^\theta s = \theta s = \mu_0 s$. 

We have $\ker \bar{\partial}^{\nabla^\theta} = \ker (\partial^{\nabla^{\theta,-1}})^* = \coker \partial^{\nabla^{\theta,-1}}$, and so there is a direct sum decomposition
\begin{equation*}
\Omega^0(\tilde{\mathcal{E}}) \cong \ker \bar{\partial}^{\nabla^\theta} \oplus \overline{\im \partial^{\nabla^{\theta,-1}}} .
\end{equation*}

First consider the case where $s \in \im \partial^{\nabla^{\theta,-1}}$ (we will generalise this to $s \in \overline{\im \partial^{\nabla^{\theta,-1}}}$ below). Write $s = \partial^{\nabla^{\theta,-1}} t$. We can make this choice unique by choosing $t \in (\ker \partial^{\nabla^{\theta,-1}})^\perp$, although this is not necessary in the following proof. Equation \eqref{eqn:neg-degree-commutativity} shows that
\begin{equation*}
\Delta_{\bar{\partial}}^\theta s = \Delta_{\bar{\partial}}^\theta \circ \partial^{\nabla^{\theta,-1}} t = i \partial^{\nabla^{\theta,-1}} * F_{\nabla^{\theta,-1}} t + \partial^{\nabla^{\theta,-1}} \circ \Delta_{\bar{\partial}}^{\theta,-1} t \in \im \partial^{\nabla^{\theta,-1}} .
\end{equation*}
Since the curvature $i * F_{\nabla^{\theta,-1}}$ is constant, then $i \partial^{\nabla^{\theta,-1}} * F_{\nabla^{\theta,-1}} t = (i*F_{\nabla^{\theta,-1}}) \partial^{\nabla^{\theta,-1}} t = (i*F_{\nabla^{\theta,-1}}) s$, and so
\begin{equation}\label{eqn:first-image}
(\nabla^\theta)^* \nabla^\theta s = 2 \Delta_{\bar{\partial}}^\theta s + i *F_{\nabla^\theta}s = (2i * F_{\nabla^{\theta,-1}} + i * F_{\nabla^\theta}) s + 2 \partial^{\nabla^{\theta,-1}} \circ \Delta_{\bar{\partial}}^{\theta,-1} t .
\end{equation}
If $t \in \ker \bar{\partial}^{\nabla^{\theta,-1}}$, then the remaining term vanishes and we obtain 
\begin{equation}
(\nabla^\theta)^* \nabla^\theta s = 2 \Delta_{\bar{\partial}}^\theta s + i *F_{\nabla^\theta}s = (2i * F_{\nabla^{\theta,-1}} + i * F_{\nabla^\theta}) s ,
\end{equation}
which gives us the bound that we want, since $i*F_{\nabla^{\theta,-1}} = \theta - (2g-2)$ is positive. 

If $t \in (\ker \bar{\partial}^{\nabla^{\theta,-1}})^\perp$, then we want to show that the final term in \eqref{eqn:first-image} is non-negative. Take the $L^2$ inner product of both sides of \eqref{eqn:first-image} with $s = \partial^{\nabla^{\theta,-1}} t$ to obtain
\begin{equation}\label{eqn:laplacian-inner-product}
\left< (\nabla^\theta)^* \nabla^\theta s , s \right> = \left< (2i * F_{\nabla^{\theta,-1}} + i * F_{\nabla^\theta}) s, s \right> + 2 \left< \partial^{\nabla^{\theta,-1}} \circ \Delta_{\bar{\partial}}^{\theta,-1} t , \partial^{\nabla^{\theta,-1}} t \right> .
\end{equation}
The final term is 
\begin{align*}
\left< \partial^{\nabla^{\theta,-1}} \circ \Delta_{\bar{\partial}}^{\theta,-1} t , \partial^{\nabla^{\theta,-1}} t \right> & = \left< \Delta_\partial^{\theta,-1} \circ \Delta_{\bar{\partial}}^{\theta,-1} t, t \right> \\
 & = \left< \left( \Delta_{\bar{\partial}}^{\theta,-1} + i * F_{\nabla^{\theta,-1}} \right) \circ \Delta_{\bar{\partial}}^{\theta,-1} t, t \right> \quad \text{from \eqref{eqn:decompose-laplacian}} \\
 & = \left< \Delta_{\bar{\partial}}^{\theta,-1} t, \Delta_{\bar{\partial}}^{\theta,-1} t \right> + i * F_{\nabla^{\theta,-1}} \left< \Delta_{\bar{\partial}}^{\theta,-1} t, t \right> .
\end{align*}
Since $i * F_{\nabla^{\theta,-1}} > 0$ and $\Delta_{\bar{\partial}}^{\theta,-1}$ is non-negative, then this final term from \eqref{eqn:laplacian-inner-product} is non-negative. Therefore \eqref{eqn:laplacian-inner-product} becomes
\begin{align*}
\left< (\nabla^\theta)^* \nabla^\theta s , s \right> & \geq \left< (2i * F_{\nabla^{\theta,-1}} + i * F_{\nabla^\theta}) s, s \right> \\
 & = (2i * F_{\nabla^{\theta,-1}} + i * F_{\nabla^\theta} ) \left< s, s \right> \\
\Leftrightarrow \quad \frac{\left< (\nabla^\theta)^* \nabla^\theta s , s \right>}{\left< s, s \right>} & \geq 2i * F_{\nabla^{\theta,-1}} + i * F_{\nabla^\theta} ,
\end{align*}
and so there is a gap in the spectrum at $\theta = i * F_{\nabla^\theta}$, since $2i * F_{\nabla^{\theta,-1}} > 0$.

Now consider a general $s \in \overline{\im \partial^{\nabla^{\theta,-1}}}$, and write $s = \lim_{n \rightarrow \infty} s_n$ for a sequence $\{ s_n \} \subset \im \partial^{\nabla^{\theta,-1}}$ such that $\| (\nabla^\theta)^* \nabla^\theta s_n \|$ is bounded and $\left\| s_n \right\|^2 = \left\| s \right\|^2$. As before, write $s_n = \partial^{\nabla^{\theta,-1}} t_n$ for each $n$ and apply the same argument as above to show that 
\begin{equation*}
\frac{\left< (\nabla^\theta)^* \nabla^\theta s_n , s_n \right>}{\left< s, s \right>} = \frac{\left< (\nabla^\theta)^* \nabla^\theta s_n , s_n \right>}{\left< s_n, s_n \right>} \geq 2i * F_{\nabla^{\theta,-1}} + i * F_{\nabla^\theta} .
\end{equation*}
Therefore, since the operator $(\nabla^\theta)^* \nabla^\theta$ is continuous and $\| (\nabla^\theta)^* \nabla^\theta s_n \|$ is bounded, then
\begin{equation*}
\frac{\left< (\nabla^\theta)^* \nabla^\theta s , s \right>}{\left< s, s \right>} = \lim_{n \rightarrow \infty} \frac{\left< (\nabla^\theta)^* \nabla^\theta s_n , s_n \right>}{\left< s, s \right>} \geq 2i * F_{\nabla^{\theta,-1}} + i * F_{\nabla^\theta} .
\end{equation*}
In summary, if $s \in \overline{\im \partial^{\nabla^{\theta,-1}}}$, then $\frac{\left< (\nabla^\theta)^* \nabla^\theta s , s \right>}{\left< s, s \right>} \geq 2i * F_{\nabla^{\theta,-1}} + i * F_{\nabla^\theta}$, and so the eigenvalues of $(\nabla^\theta)^* \nabla^\theta$ on $\overline{\im \partial^{\nabla^{\theta,-1}}}$ are bounded below by $2i * F_{\nabla^{\theta,-1}} + i * F_{\nabla^\theta}$. Since $2i * F_{\nabla^{\theta,-1}} > 0$ then there is a spectral gap at the lowest eigenvalue $i * F_{\nabla^\theta}$.

Continuing inductively, suppose that $s \in \overline{\im \partial^{\nabla^{\theta,-1}} \circ \cdots \circ \partial^{\nabla^{\theta,-q}}}$. Again, we can decompose $\Omega^0(\tilde{\mathcal{E}} \otimes K^{-q}) \cong \overline{\im \partial^{\nabla^{\theta,-(q+1)}}} \oplus \overline{\coker \partial^{\nabla^{\theta,-(q+1)}}}$. Given any $t \in \Omega^0(\tilde{\mathcal{E}} \otimes K^{-q})$, there exists a sequence $\{t_n\} \subset \im \partial^{\nabla^{\theta,-(q+1)}} \oplus \coker \partial^{\nabla^{\theta,-(q+1)}}$ such that $\lim_{n \rightarrow \infty} t_n = t$. Write $t_n = \partial^{\nabla^{\theta,-(q+1)}} u_n + v_n$ with $u_n \in \Omega^0(\tilde{\mathcal{E}} \otimes K^{-(q+1)})$ and $v_n \in \coker \partial^{\nabla^{\theta,-(q+1)}} \cong \ker \bar{\partial}^{\nabla^{\theta,-q}}$. 

The same argument as before shows that
\begin{align*}
\Delta_{\bar{\partial}}^\theta (\partial^{\nabla^{\theta,-1}} \circ \cdots \circ \partial^{\nabla^{\theta,-q}} v_n) & = \partial^{\nabla^{\theta,-1}} \circ \cdots \circ \partial^{\nabla^{\theta,-q}} \circ \Delta_{\bar{\partial}}^{\theta,-q} v_n \\
 & \quad \quad + \left( \sum_{\ell=1}^q (\theta - \ell (2g-2)) \right) \partial^{\nabla^{\theta,-1}} \circ \cdots \circ \partial^{\nabla^{\theta,-q}} v_n \\
 & = \left( q \theta - \frac{q(q+1)(2g-2)}{2}  \right) \partial^{\nabla^{\theta,-1}} \circ \cdots \circ \partial^{\nabla^{\theta,-q}} v_n ,
\end{align*}
and therefore $\partial^{\nabla^{\theta,-1}} \circ \cdots \circ \partial^{\nabla^{\theta,-q}} v_n$ is an eigensection of $(\nabla^\theta)^* \nabla^\theta = 2 \Delta_{\bar{\partial}}^\theta + i * F_{\nabla^\theta}$ with eigenvalue $2\left( q \theta - \frac{q(q+1)(2g-2)}{2} \right) + \theta = (2q+1) \theta - q(q+1)(2g-2)$.

Now consider $s_n =  \partial^{\nabla^{\theta,-1}} \circ \cdots \circ \partial^{\nabla^{\theta,-(q+1)}} u_n$. Again, the same argument as above shows that
\begin{equation}\label{eqn:laplacian-stage-q}
\Delta_{\bar{\partial}}^\theta s_n = \partial^{\nabla^{\theta,-1}} \circ \cdots \circ \partial^{\nabla^{\theta,-(q+1)}} \Delta_{\bar{\partial}}^{\theta,-(q+1)} u_n + \sum_{\ell=1}^{q+1} \left( \theta - \ell (2g-2) \right) s_n .
\end{equation}
It remains to show that the inner product of the first term with $s_n = \partial^{\nabla^{\theta,-1}} \circ \cdots \circ \partial^{\nabla^{\theta,-(q+1)}} u_n$ is non-negative. To see this, write
\begin{multline*}
\left< \partial^{\nabla^{\theta,-1}} \circ \cdots \circ \partial^{\nabla^{\theta,-(q+1)}} \Delta_{\bar{\partial}}^{\theta,-(q+1)} u_n , \partial^{\nabla^{\theta,-1}} \circ \cdots \circ \partial^{\nabla^{\theta,-(q+1)}} u_n \right> \\
 = \left< (\partial^{\nabla^{\theta,-1}})^* \partial^{\nabla^{\theta,-1}} \circ \cdots \circ \partial^{\nabla^{\theta,-(q+1)}} \Delta_{\bar{\partial}}^{\theta,-(q+1)} u_n , \partial^{\nabla^{\theta,-2}} \circ \cdots \circ \partial^{\nabla^{\theta,-(q+1)}} u_n \right> \\
 = \left< \Delta_\partial^{\theta,-1} \partial^{\nabla^{\theta,-2}} \circ \cdots \circ \partial^{\nabla^{\theta,-(q+1)}} \Delta_{\bar{\partial}}^{\theta,-(q+1)} u_n , \partial^{\nabla^{\theta,-2}} \circ \cdots \circ \partial^{\nabla^{\theta,-(q+1)}} u_n \right> \\
 = \left< \left( i *F_{\nabla^{\theta,-1}} + \Delta_{\bar{\partial}}^{\theta,-1} \right) \partial^{\nabla^{\theta,-2}} \circ \cdots \circ \partial^{\nabla^{\theta,-(q+1)}} \Delta_{\bar{\partial}}^{\theta,-(q+1)} u_n , \partial^{\nabla^{\theta,-2}} \circ \cdots \circ \partial^{\nabla^{\theta,-(q+1)}} u_n \right> \\
 = (i*F_{\nabla^{\theta,-1}}) \left< \partial^{\nabla^{\theta,-2}} \circ \cdots \circ \partial^{\nabla^{\theta,-(q+1)}} \Delta_{\bar{\partial}}^{\theta,-(q+1)} u_n , \partial^{\nabla^{\theta,-2}} \circ \cdots \circ \partial^{\nabla^{\theta,-(q+1)}} u_n \right> \\
  + \left< \partial^{\nabla^{\theta,-2}} \circ \cdots \circ \partial^{\nabla^{\theta,-(q+1)}} \Delta_{\bar{\partial}}^{\theta,-(q+1)} u_n , \Delta_{\bar{\partial}}^{\theta,-1} \partial^{\nabla^{\theta,-2}} \circ \cdots \circ \partial^{\nabla^{\theta,-(q+1)}} u_n \right> .
\end{multline*}
In the final expression above, the first term is a positive multiple of 
\begin{equation*}
\left< \partial^{\nabla^{\theta,-2}} \circ \cdots \circ \partial^{\nabla^{\theta,-(q+1)}} \Delta_{\bar{\partial}}^{\theta,-(q+1)} u_n , \partial^{\nabla^{\theta,-2}} \circ \cdots \circ \partial^{\nabla^{\theta,-(q+1)}} u_n \right> .
\end{equation*}
In order to evaluate the second term, we can use Lemma \ref{lem:failure-to-commute} and the same process as in the proof of Proposition \ref{prop:eigensections} to write
\begin{align*}
\Delta_{\bar{\partial}}^{\theta,-1} \partial^{\nabla^{\theta,-2}} \circ \cdots \circ \partial^{\nabla^{\theta,-(q+1)}} u_n & = \partial^{\nabla^{\theta,-2}} \circ \cdots \circ \partial^{\nabla^{\theta,-(q+1)}} \Delta_{\bar{\partial}}^{\theta,-(q+1)} u_n \\
 & \quad \quad + C \cdot \partial^{\nabla^{\theta,-2}} \circ \cdots \circ \partial^{\nabla^{\theta,-(q+1)}} u_n ,
\end{align*}
where $C$ is a positive constant. Therefore we obtain a positive multiple of 
\begin{equation*}
\left< \partial^{\nabla^{\theta,-2}} \circ \cdots \circ \partial^{\nabla^{\theta,-(q+1)}} \Delta_{\bar{\partial}}^{\theta,-(q+1)} u_n , \partial^{\nabla^{\theta,-2}} \circ \cdots \circ \partial^{\nabla^{\theta,-(q+1)}} u_n \right>
\end{equation*}
plus the non-negative term $\left< \partial^{\nabla^{\theta,-2}} \circ \cdots \circ \partial^{\nabla^{\theta,-(q+1)}} \Delta_{\bar{\partial}}^{\theta,-(q+1)} u_n , \partial^{\nabla^{\theta,-2}} \circ \cdots \circ \partial^{\nabla^{\theta,-(q+1)}} \Delta_{\bar{\partial}}^{\theta,-(q+1)} u_n \right>$.

Now repeat the same process on the term 
\begin{equation*}
\left< \partial^{\nabla^{\theta,-2}} \circ \cdots \circ \partial^{\nabla^{\theta,-(q+1)}} \Delta_{\bar{\partial}}^{\theta,-(q+1)} u_n , \partial^{\nabla^{\theta,-2}} \circ \cdots \circ \partial^{\nabla^{\theta,-(q+1)}} u_n \right> .
\end{equation*}
Continuing inductively, at the end of the process we obtain non-negative terms plus a positive multiple of the non-negative term $\left< \Delta_{\bar{\partial}}^{\theta,-(q+1)} u_n, u_n \right>$. Substituting this into \eqref{eqn:laplacian-stage-q}, we see that
\begin{equation*}
\left< \Delta_{\bar{\partial}}^\theta s_n , s_n \right> \geq \sum_{\ell=1}^{q+1} \left( \theta - \ell (2g-2) \right) \left< s_n, s_n \right> .
\end{equation*}
By taking sequences as before, this inequality is preserved in the limit, and therefore it applies to sections in the closure. Since $\theta - m (2g-2) > 0$ by assumption, then the final term in this sum is strictly positive if $q \leq m-1$, and so there is a gap in the spectrum of $\Delta_{\bar{\partial}}^\theta$ at the eigenvalue $\sum_{\ell=1}^q (\theta - \ell (2g-2) )$. By \eqref{eqn:Laplacian-curvature} there is then a gap in the spectrum of $(\nabla^\theta)^* \nabla^\theta$ at the eigenvalue $\mu_q$ when $q \leq m-1$.

Finally, since $\partial^{\nabla^{\theta,-\ell}}$ commutes with $T_\gamma$ for all $\ell$ by the results of Section \ref{subsec:magnetic-schrodinger} and it is injective for all $\ell \leq m$ by Remark \ref{rem:injective-10}, then this $\partial^{\nabla^{\theta,-1}} \circ \cdots \circ \partial^{\nabla^{\theta,-q}}$ defines an isomorphism of $A_\theta$ modules
\begin{equation*}
H^0(\tilde{\mathcal{E}} \otimes K^{-q}) \cong E_{\mu_q} = \overline{\partial^{\nabla^{\theta,-1}} \circ \cdots \circ \partial^{\nabla^{\theta,-q}} \left( \ker \bar{\partial}^{\nabla^{\theta,-q}} \right)} . 
\end{equation*}
\end{proof}

 \section{An Index Theorem}
 
In the previous notation, let  $\nabla^\theta $ be the unitary connection on the trivial line bundle on $Z$, whose curvature is equal to $i \tilde B=i \theta \tilde\omega$. Consider the Dolbeault operator $\dbar$ on $X$ , and lift it to $\tilde\dbar$ on $Z$, which is 
 invariant under the group of deck transformations $\ZZ^{2g}$ for the abelian cover $Z \rightarrow X$. Let $\cE\to X$ be a holomorphic vector bundle over $X$, and $\nabla^\cE$ a $(0,1)$ connection on $\cE$. Let $\tilde\cE$ denote the lift of $\cE$ to $Z$, and $\tilde\nabla^\cE$ the lift of $\nabla^\cE$ to $Z$. Then as before, one checks that the operator $\tilde\dbar\otimes\tilde\nabla^\cE\otimes \nabla^\theta$ is invariant under the projective action of $\ZZ^{2g}$ as in the
 previous section.  Here
 \begin{equation}\label{eqn:L2-complex}
 \tilde\dbar\otimes\tilde\nabla^\cE\otimes \nabla^\theta: \Omega_{(2)}^{0,0}(Z, \cE) \longrightarrow \Omega_{(2)}^{0,1}(Z, \cE) 
 \end{equation}
 where $\Omega_{(2)}^{0,j}(Z, \cE), \, j=0,1$ denotes the space of square integrable differential $j$-forms on $Z$ with coefficients in $\cE$.

Define 
\begin{equation*}
h^0(\tilde{\nabla}^{\cE}\otimes \nabla^\theta) = \dim_\tau(\ker_{L^2}(\tilde\dbar\otimes \tilde{\nabla}^{\cE}\otimes \nabla^\theta) ) .
\end{equation*}
By $L^2$-Serre duality,
\begin{equation*}
h^1(\tilde\nabla^\cE\otimes \nabla^\theta)  = h^0(\tilde{\nabla}^{\cE^* \otimes K} \otimes \nabla^{-\theta}) = \dim_\tau(\ker_{L^2} ( \tilde\dbar\otimes \tilde{\nabla}^{\cE^* \otimes K} \otimes \nabla^{-\theta})).
\end{equation*}
where $\tau$ denotes the von Neumann trace or dimension. Elliptic regularity ensures that these numbers are finite.

The goal of this section is to prove the following result.

\begin{theorem}[$L^2$-Riemann-Roch for projective actions]\label{thm:RR}
The $L^2$ index of the induced connection $\tilde \dbar \otimes \nabla^{K^{-j}}\otimes \nabla^\theta$ is
\begin{align}\label{eqn:L2-index}
\begin{split}
{ \rm index}_{L^2}( \tilde\dbar\otimes \nabla^{\mathcal{E}}\otimes \nabla^\theta) & = h^0(\nabla^{\mathcal{E}}\otimes \nabla^\theta) - h^1(\nabla^{\mathcal{E}}\otimes \nabla^\theta) \\
 & = \deg(\cE)+rk(\cE) \theta + (1-g) rk(\cE)
\end{split}
\end{align}
 In the case where $\mathcal{E} = K_X^{-q}$ is a power of the canonical bundle on $X$, then
\begin{enumerate}
\item the $L^2$ index is
\begin{equation*}
h^0(\nabla^{K^{-q}} \otimes \nabla^\theta) - h^1(\nabla^{K^{-q}}\otimes \nabla^\theta)  = (2 q + 1) (1-g) + \theta .
\end{equation*}
 \item (vanishing) If $\theta > 2g-2,$ then $h^1(\nabla^\theta) = 0$, so that $h^0(\nabla^\theta) = 1-g + \theta$.
 \end{enumerate}
 \end{theorem}

\begin{proof} 
 Note that since $Z$ is a Riemann surface, then \eqref{eqn:L2-complex} is a complex, and  
 by the higher twisted index theorem in \cite{MathaiMarcolli01}, the index is
 $$
 {\rm Index}_{A_\theta}(\tilde\dbar\otimes\tilde\nabla^\cE\otimes \nabla^\theta) = [\ker( \tilde\dbar\otimes\tilde\nabla^\cE\otimes \nabla^\theta) -
 \coker( \tilde\dbar\otimes\tilde\nabla^\cE\otimes \nabla^\theta)] \in K_0(\cA_\Theta) .
 $$
 As in the index for a family of elliptic operators, in general neither $\coker( \tilde\dbar\otimes\tilde\nabla^\cE\otimes \nabla^\theta)$ nor 
 $\ker( \tilde\dbar\otimes\tilde\nabla^\cE\otimes \nabla^\theta)$ are finite projective modules, however under a vanishing condition $\coker( \tilde\dbar\otimes\tilde\nabla^\cE\otimes \nabla^\theta) =0$, then $\ker( \tilde\dbar\otimes\tilde\nabla^\cE\otimes \nabla^\theta)$
 is a finite projective module, which is therefore a holomorphic vector bundle in terminology of subsection \ref{complexnctori}.

Then by the twisted index theorem \cite{Mathai99}, one has
\begin{align*}
\tau\left( {\rm Index}_{A_\theta}(\tilde\dbar\otimes\tilde\nabla^\cE\otimes \nabla^\theta)\right)
&= \frac{1}{2\pi}\int_X \text{Todd}(\Omega_X) \tr\exp(\Omega^\cE) \exp(\theta \omega)\\
&=  \frac{1}{2\pi}\int_X (1+\frac{1}{2}\Omega_X) (rk(\cE)+ \tr(\Omega^\cE)) (1+ \theta \omega)\\
&=  \text{deg}(\cE)+rk(\cE) \theta + (1-g) rk(\cE).
 \end{align*}
 On the other hand,
 $$
 \tau\left( {\rm Index}_{A_\theta}(\tilde\dbar\otimes\tilde\nabla^\cE\otimes \nabla^\theta)\right)
 = h^0(\tilde\nabla^\cE\otimes \nabla^\theta) -h^1(\tilde\nabla^\cE\otimes \nabla^\theta) ,
 $$
 which completes the proof of \eqref{eqn:L2-index}.

Now let $\cE= K_X^{-q}$ be the holomorphic line bundle. We conclude from  \eqref{eqn:L2-index} that 
\begin{equation}
h^0(\tilde\nabla^{K^{-q}}\otimes \nabla^\theta)= (-q \deg(K_X) + \theta) + 1-g = (2q+1)(1-g) + \theta > 0.
\end{equation}
If $q\ge 0$ is such that $(q+1)(2g-2) - \theta < 0$, then \eqref{eqn:decompose-laplacian} implies that $h^0(\tilde\nabla^{K^{q+1}} \otimes \nabla^{-\theta}) = 0$. Using $L^2$ Serre duality we conclude that
\begin{equation}\label{eqn:vanishing1}
h^1(\tilde\nabla^{K^{-q}} \otimes \nabla^\theta) = h^0(\tilde\nabla^{K^{q+1}} \otimes \nabla^{-\theta}) = 0 .
\end{equation}
\end{proof}

Theorem \ref{thm:discrete-spectrum} shows that $\ker_{A_\theta}(\tilde\dbar\otimes\tilde\nabla^\cE\otimes \nabla^\theta) = E_{\mu_q}$, the $\mu_q$-th eigenspace of the magnetic Laplacian. Using Theorem \ref{thm:RR} and the terminology in subsection \ref{complexnctori}, we have the following result on the dimension of these eigenspaces, which proves part \eqref{item:von-Neumann-dimension} of Theorem \ref{thm:main-results}.
 
 \begin{lemma}\label{dim}
Let $E_{\mu_q}$ be the $\mu_q$-th eigenspace of the magnetic Laplacian. Then
\begin{enumerate}
\item  $E_{\mu_q}$ is a holomorphic vector bundle over $\cA_\theta\otimes \cK(L^2(\cF))$, and 
\item The von Neumann dimension is $\dim_\tau(E_{\mu_q}) = (2q+1) (1-g)+ \theta>0$.
\end{enumerate}
\end{lemma}

\section{Chern number of $E_{\mu_q}$ via a higher index theorem}

Our goal is to calculate the Chern number of the spectral subspace $E_{\mu_q}$ as a finite projective $A_\theta$-module, which will complete the proof of the final part of Theorem \ref{thm:main-results}. We do this by using the higher twisted index theorem \cite[Thm. 2.2]{MathaiMarcolli01}. The argument goes as follows. Let $c$ be the group cocycle on $\ZZ^{2g}$ corresponding to the symplectic 2-form, and let $\tau_c$ be the corresponding continuous 2-cocycle from \eqref{eqn:continuous-cocycle-def}. Applying the higher twisted index theorem and simplifying, we see that
$$
\tau_c({ \rm index}_{A_\theta}( \tilde\dbar\otimes \tilde{\nabla}^{K^{-j}}\otimes \nabla^\theta)) = \int_X \phi(c).
$$
By the vanishing theorem \eqref{eqn:vanishing1}, we deduce that if $\theta > (j+1)(2g-2)$ then ${ \rm ker}_{A_\theta}( \tilde\dbar\otimes \nabla^{K^{-j}}\otimes \nabla^\theta)$ is a finite projective $A_\theta$ module, and 
$$
\tau_c({ \rm ker}_{A_\theta}( \tilde\dbar\otimes \tilde{\nabla}^{K^{-j}}\otimes \nabla^\theta)) = \int_X \phi(c).
$$
Now $\phi(c)$ is easily seen to be the Bergman volume form on $X$, and  $\int_X \phi(c) = 2g$ if $X$ is a genus $g$ Riemann surface. Finally, we will prove in Theorem \ref{thm:discrete-spectrum} that 
$$
{ \rm ker}_{A_\theta}( \tilde\dbar\otimes \tilde{\nabla}^{K^{-j}}\otimes \nabla^\theta) \cong E_{\mu_q}
$$
as $A_\theta$ modules, therefore we conclude that if $X$ is a genus $g$ Riemann surface, then the Chern number of $E_{\mu_q}$ is 
\begin{equation}\label{eqn:chern-number-smooth-proof}
\tau_c(E_{\mu_q}) = \int_\Sigma \phi(c) = 2g,
\end{equation}
which proves \eqref{eqn:chern-number-eigenspace}.

We can extend the above calculations to the case of a ``good'' orbifold given by a quotient $X = X' / \Gamma$ as follows. Let $g'$ be the genus of $X'$ and let $p : X' \rightarrow X$ be the associated ramified cover. Since $\phi(c)$ pulls back to the Bergman volume form on $X'$, then 
\begin{equation*}
\int_X \phi(c) = \frac{1}{|\Gamma|} \int_{\Sigma'} p^* \phi(c) = \frac{2g'}{|\Gamma|} .
\end{equation*} 
Let $n_y = |\Gamma_y|$ be the order of the isotropy group at each point $y \in X'$, let $R = \{ y \in X' \, : \, n_y > 1 \}$ be the ramification divisor, and let $n = |R|$ be the total number of ramification points in $X'$. The Riemann-Hurwitz theorem shows that
\begin{align*}
2g'-2 & = |\Gamma| (2g-2) + \sum_{y \in X'} (n_y - 1) \\
\Leftrightarrow \quad \frac{2g'}{|\Gamma|} & = 2g-2 + \frac{2}{|\Gamma|} + \frac{1}{|\Gamma|} \sum_{y \in X'} (n_y - 1) \\
 & = 2g-2 + \frac{1}{|\Gamma|} \sum_{y \in R} n_y + \frac{2-n}{|\Gamma|} .
\end{align*}
The term $\frac{1}{|\Gamma|} \sum_{y \in R} n_y$ can be further simplified by noting that for each $y \in R$, the number of points in the orbit $\Gamma \cdot y$ is equal to $\frac{|\Gamma|}{n_y}$. Therefore $\frac{1}{|\Gamma|} \sum_{y \in R} n_y$ is equal to the number of $\Gamma$-orbits in $R$, which we denote by $\# (R / \Gamma)$. In conclusion, we have proved that
\begin{equation*}
\int_X \phi(c) = 2g-2 + \# (R / \Gamma) + \frac{2-n}{|\Gamma|} .
\end{equation*}
By Theorem \ref{thm:discrete-spectrum}, we know that there is an isomorphism of $A_\theta$ modules $\ker_{A_\theta}(\tilde\dbar\otimes\tilde\nabla^\cE\otimes \nabla^\theta) \cong E_{\mu_q}$, the $\mu_q$-th eigenspace of the magnetic Laplacian. All of the above is summarised in the following lemma, which proves \eqref{eqn:chern-number-orbifold} and thus concludes the proof of Theorem \ref{thm:main-results}.
 
 \begin{lemma}
The Chern class of the $\mu_q$-th eigenspace of the magnetic Laplacian is
\begin{equation}\label{eqn:chern-number-orbifold-proof}
\tau_c(E_{\mu_q}) = 2g-2 + \# (R / \Gamma) + \frac{2-n}{|\Gamma|} .
\end{equation}
 \end{lemma}


\end{document}